\documentclass[conference]{IEEEtran}
\IEEEoverridecommandlockouts
\usepackage{cite}
\usepackage{amsmath,amssymb,amsfonts}

\usepackage{algorithmic}
\usepackage{algorithm2e}
\usepackage{graphicx}
\usepackage{textcomp}

\def\BibTeX{{\rm B\kern-.05em{\sc i\kern-.025em b}\kern-.08em
    T\kern-.1667em\lower.7ex\hbox{E}\kern-.125emX}}

\usepackage{bbm}
\usepackage{mathrsfs} 
\usepackage[table]{xcolor}
\usepackage{nicefrac}       
\usepackage{microtype}      
\usepackage{multirow}
\usepackage{graphicx}
\usepackage{subcaption}
\usepackage{breqn}
\usepackage{caption}
\usepackage{bbm}

\newtheorem{prop}{Proposition}

\newcommand{\qedsymbol}{\rule{0.7em}{0.7em}}
\newcommand{\xhdr}[1]{\vspace{0.1mm}\noindent{{\bf #1.}}}
\renewcommand{\a}{\alpha }
\newcommand{\sign}{\mathop{\bf sign}}
\newcommand{\vect}[1]{\boldsymbol{#1}}

\newcommand{\pt}{\forall }

\newcommand{\q}{\quad}
\newcommand{\oo}{\infty}

\newcommand{\R}{{\mathbb R}}

\usepackage{bbm}

\newcommand{\pro}{\begin{proposition}}

\usepackage{microtype}
\usepackage{graphicx}
\usepackage{booktabs} 

\usepackage{hyperref}

\begin{document}
    
\title{Convex Hierarchical Clustering for Graph-Structured Data
}

\author{\IEEEauthorblockN{Claire Donnat}
\IEEEauthorblockA{\textit{Department of Statistics} \\
\textit{Stanford University}\\
Stanford, CA, USA \\
cdonnat@stanford.edu}
\and
\IEEEauthorblockN{Susan Holmes}
\IEEEauthorblockA{\textit{Department of Statistics} \\
\textit{Stanford University}\\
Stanford, CA, USA \\
susan@stat.stanford.edu}
}

\maketitle




\begin{abstract}
Convex clustering \cite{chi2017convex} is a recent stable alternative to hierarchical clustering. It formulates the recovery of progressively coalescing clusters as a regularized convex problem. While convex clustering was originally designed for handling Euclidean distances between data points, in a growing number of applications, the data is directly characterized by a similarity matrix or weighted graph. In this paper, we  extend the robust hierarchical clustering approach to these broader classes of similarities. Having defined an appropriate convex objective, the crux of this adaptation lies in our ability to provide: (a) an efficient recovery of the regularization path and (b) an empirical  demonstration of the use of our method.  We address the first challenge through a proximal dual algorithm, for which we characterize both the theoretical efficiency as well as the empirical performance on a set of experiments. Finally, we highlight the potential of our method by showing its application to several real-life datasets, thus providing a natural extension to the current scope of applications of convex clustering.
\end{abstract}

\section{Introduction and related work}


From gene sequencing to biomedical studies \cite{chipman2005hybrid,corpet1988multiple,girvan2002community,segal2002probabilistic}, hierarchical clustering \cite{day1984efficient,johnson1967hierarchical} is currently one of the most widely-used procedures for data analysis. The dendrogram that it provides yields a complete summary of the data and bypasses the need to prespecify an adequate number of clusters. The visualization of these clusters' progressive coalescence provides a comprehensive and intuitive view of their similarities. However, hierarchical clustering is an inherently greedy procedure, which typically constructs the clusters' fusion path by iteratively aggregating (or splitting) clusters. The recovered coalescence path is dependent on the choice of the linkage function, and has also been shown to be highly sensitive to outliers and perturbations of the dataset--- thus allowing the formation of spurious clusters and consequently hindering the generalizability of the analysis \cite{chen2015convex}. This is particularly problematic in applications where such multiscale representations of the data are compared, contrasted and analyzed. Such is the case in brain connectomics, where a topic of interest  is the comparison of the multiscale representations of the network woven by white matter tracts across different people or groups. In such noisy regimes, the definition of a robust and optimal hierarchical clustering takes on a particular importance.  \\
\vspace{-0.3cm}

\xhdr{Convex clustering} To overcome these issues, convex clustering \cite{chen2015convex,hocking2011clusterpath,pelckmans2005convex} is a recent alternative formulation of hierarchical clustering as the solution of a convex optimization problem with a regularization penalty. In its original form, denoting each observation $i$ by its corresponding vector $X_i \in \R^{d} (\text{with } i=1 \cdots N$), and introducing $U_i \in \R^{d}$ the centroid of the cluster associated to point $i$, convex clustering solves the following objective:
\vspace{-2mm}
\begin{equation}\label{eq:conv_clust}
	\underset{U \in \R^{d \times N }}{\text{argmin}}\sum_{i=1}^N || X_i -U_i||^2 +\lambda \sum_{i,j=1}^N W_{ij} \text{Pen}(U_i -U_j)
\end{equation}
where $W_{ij}$ are coupling weights (typically chosen as the $k$-nearest neighbors of each observation). In the previous expression, $\textit{Pen}$ is a penalty function (typically the $\ell_q$-norm, with $q \geq 1$) which encourages coupled observations to share the same centroid. The solution path $U^{(\lambda)}$ is comparable to the coalescence path recovered by hierarchical clustering, and the regularization parameter $\lambda$, to the different levels in the hierarchical clustering dendrogram\cite{chen2015convex,hocking2011clusterpath}:
\begin{itemize}
\itemsep0em
	\item for $\lambda=0$, the solution of Eq. \ref{eq:conv_clust} is $U^{(0)}=X$, and each point belongs its own cluster.
	\item as $\lambda$ increases, the penalty term induces the centroids $U_i$ to fuse, until in the limit, these centroids reach a consensus value, thus forming a single cluster $U^{(\oo)}: \forall i,j \in \{1, N\}, \q U_i^{(\oo)}=U_j^{(\oo)}$.
\end{itemize}
The strict convexity of the objective function of Problem \ref{eq:conv_clust} guarantees the existence of a globally optimal solution, as well as its robustness against perturbations \cite{chen2015convex}, thus making this convex formulation an extremely appealing alternative to hierarchical clustering. We  refer the reader to \cite{tan2015statistical} for a thorough review of the properties of convex clustering as well as a formal analysis of its parallel with hierarchical clustering.\\
\vspace{-3mm}

\xhdr{Contributions and related work} One of the main drawbacks of convex clustering is that the optimization procedure associated to Problem \ref{eq:conv_clust} is computationally more involved than the greedy optimization performed by standard hierarchical clustering. While some work has already been put into the design of efficient solutions \cite{chi2017convex}, to the best of our knowledge, the derivation of algorithms for convex clustering has been restricted to the setting where data are Euclidean: observations are represented by vectors in $\R^d$, and similarities are simply characterized through pairwise Euclidean distances. However, in an increasing number of applications, such representations are difficult to obtain and the data come more readily as a graph in which nodes represent observations, and edges reflect some function of similarities between data points. In connectomics for instance, correlations between brain regions are summarized by a weighted graph, which provides a more amenable support to the study of functional connectivity \cite{kelly2011reduced,tetreault2016brain}. Similarly, in social sciences, relationships between individuals are readily modeled by a graph, where edges denote interactions between users.  
In many of these graph-structured datasets, hierarchical clustering is an indispensable tool since it allows the analysis of the data at different scales. The derivation of a convex multiscale summary of the data with the same global optimum and robustness guarantees as its Euclidean counterpart thus represents an impactful problem with many applications --  a challenge which we propose to tackle in this paper. Our contributions consist in (a) the adaptation of the convex objective posed in Problem \ref{eq:conv_clust} to the graph setting and (b) the derivation of two provably efficient solutions. We analyze and validate our method through a set of synthetic experiments, and show its application on several real-world datasets.

\section{Problem Statement}

Throughout this paper, we assume that the data comes under the form of a weighted similarity matrix $K$ between $N$ elements (i.e, for instance, the adjacency matrix or diffusion map associated to a graph), which we assume to be sparse. We adopt the standard convention of referring to the $i^{th}$ column of any given matrix $M$ as $M_i$.

\xhdr{Positive Definite Symmetric Input Matrices} We begin by studying the case where the similarity matrix $K$ is symmetric and Positive Definite.  By direct application of the spectral lemma, we  can re-write $K$ as a dot-product in a higher-dimensional space: $K = \Phi^T \Phi \hspace{1mm} \text{ i.e. } \forall i, j, \quad  K_{ij}=\Phi(X_i)^T \Phi(X_j)$. This provides an amenable setting for  the generalization of convex clustering, where the goal becomes to recover the centroids $U_i$ associated to each implicit high-dimensional vector $\Phi(X_i)$. Since each centroid lies in the convex hull of its corresponding vectors, we require $U$ to have the form:
\vspace{-0.3cm} 
\begin{multline}\label{eq:u} 
U =\Phi(X) \pi,  \hspace{1mm} \text{ where }\hspace{1mm}  \pi \vect{1}=\vect{1},  \quad  \vect{1}^T\pi =\vect{1}^T  \text{ and } \pi \geq 0
\end{multline}
In this setting, the doubly-stochastic matrix $\pi$ benefits from a bi-dimensional interpretation: the columns correspond to the centroids' representation using the original observations as dictionary, while the rows can be interpreted as soft membership assignments of observations to clusters. \textit{However, we highlight that this constraint further adds to the complexity of the original convex clustering algorithm, and is non-trivial to implement.}\\
Using the kernel trick,  Eq. \ref{eq:conv_clust} can be adapted here to:
\begin{equation} \label{eq:sim_CC}
\begin{split}
\hspace{-0.3cm} \underset{\pi _\in \Delta_N}{\text{argmin } } \text{Tr} [ \pi^T K \pi  -2 K  \pi ] + \lambda \sum_{i,j}K_{ij}\text{Pen}(\pi_{i} -\pi_{ j})\\
\end{split}	
\end{equation}
where $\Delta_N=\Big\{\pi \in \mathbb{R}^{N \times  N}:  \pi \vect{1}=\vect{1},  \vect{1}^T\pi =\vect{1}^T, \pi \geq 0 \Big\}$ is the set of doubly stochastic matrices.
A full derivation of this formulation can be found in \ref{appendix:further_proofs}. The next important step consists in choosing the coupling penalty, which we take here to be a mixed total-variation penalty:
\vspace{-0.3cm} 
\begin{multline*}   \text{Pen}\big(\pi_{\cdot i} -\pi_{\cdot j}\big) =\alpha || \pi_i -\pi_j||_{2,1} +(1-\alpha) ||\pi_i -\pi_j||_1.  
\end{multline*}
\vspace{-0.7cm} 

This choice is motivated by the fact that the $\ell_1$-penalty is known to provide nested sequences of clusters \cite{hocking2011clusterpath}, while the $\ell_{21}$-penalty allows the recovery of a more stable solution. Total variation distances have also been shown to encourage the recovery of piecewise linear functions and to provide solutions with sharp edge contrasts \cite{chambolle2016geometric} ---  a desirable property in our setting, since this amounts to ``clamping" the centroids together as they progressively coalesce.

To ease notation, for any square matrix $M$, we denote as $\delta^{(M)} \in \R^{N \times N^2}$ the $N \times N^2$-dimensional matrix of pairwise differences such that: $ \forall i, j, k \leq N, \quad   \delta^{(M)}_{k,(i,j)}=  (\mathbf{e}_{ki} -\mathbf{e}_{kj}) M_{ij}$, where $\mathbf{e}_{ki} = \mathbbm{1}_{k=i}$ is the $i^{th}$ cartesian column-basis vector. The final constrained minimization problem can thus be compactly written as:
\vspace{-0.3cm} 
\begin{multline}  \label{eq:sim_CC_fin} 
\underset{ \pi _\in \Delta_N}{\text{argmin }} \Big\{ \text{Trace} \big[ \pi^T K \pi  -2 K  \pi\big] \\+ \lambda  \Big(   \alpha || \pi \delta^{(K)} ||_{2,1}+(1-\alpha) ||\pi \delta^{(K)} ||_1 \Big) \Big\}
\end{multline}
\vspace{-0.3cm}

As for its Euclidean counterpart, the solution of  Eq. \ref{eq:sim_CC_fin} is consistent with its interpretation as a cluster coalescence path:
\begin{itemize}
\itemsep0em
	\item when $\lambda=0$, the solution of the previous equation is the identity: $\pi^{(0)}=I_N$. It is easy to check that $I_N$ is a solution to ${\text{argmin}}_{\pi \in [0,1]^N}  \text{Tr}[\pi^TK\pi -2K\pi]$. Since $I_N$ is doubly-stochastic, by strict convexity of the objective in Eq. \ref{eq:sim_CC_fin}, we deduce that it is the solution for $\lambda=0$. 
	\item when $\lambda =\oo$, on the other hand, the solution of Eq. \ref{eq:sim_CC_fin} must be such that $|| \pi \delta^{(K)} ||=0$. This is given by the consensus matrix $\pi^{(\oo)}=\frac{1}{N} \mathbf{1}\mathbf{1}^T$, which is the intersection of the set $\{A\in \R^{N\times N}: \forall i,j,\leq N  \quad A_i =A_j, A\geq 0 \}$ with the set of doubly-stochastic matrices.
\end{itemize}

\xhdr{Discussion} The assumption that $K$ is positive definite is by no means restrictive. Indeed, in many applications (brain connectomes, etc.), the kernel $K$ corresponds to some transformation of a positive definite similarity (typically, to some thresholded-measure of the correlation between vertices). Even if this is not the case, Positive-Definiteness can be achieved by regularizing the kernel: $\hat{K} =K +\gamma I_{n}$. 
 As for many clustering algorithms (choice of the most adequate distance metric in k-means, bandwidth in spectral clustering, etc.), the choice of the appropriate transformation should depend on the analysis.\\
We also highlight that, in contrast to hierarchical clustering, only for the choice $\alpha=0$ is the algorithm proven to output a nested sequence  of clusters\cite{hocking2011clusterpath}. However, we emphasize that the goal of our paper is to extract robust multiscale representations (rather than strictly nested ones): the regularization path allows the recovery of progressively coarser and coarser representations of the data. The strict convexity of the objective in Eq. \ref{eq:sim_CC_fin} ensures its global optimality, which in this case, is potentially  a more desirable quality than nestedness.

\section{Algorithm}
\vspace{-0.1cm}
The main challenge consists in devising an efficient algorithm for solving the previous optimization problem. While this problem is strongly convex, exact solvers are extremely slow, making the computation of the full regularization path almost intractable. In this paper, we propose two methods. The first is based on an adaptation of the Fast Iterative Shrinkage and Thresholding Algorithm \cite{beck2009fastgen}, a method originally proposed by Beck and Teboulle for image deblurring in 2009 \cite{beck2009fast} and which we have selected for both its theoretical efficiency and its empirical performance.
Our contribution here lies in the adaptation of this method to the evermore-challenging setting of Eq.  \ref{eq:sim_CC_fin}, in which the optimization has to be done on the set of doubly stochastic matrices --- a much more constrained and complicated setting than for image processing. We also provide a gradient descent-based implementation, based on a linearization of the objective and more suitable to the analysis of larger graphs, as well as an ADMM-based implementation \cite{admm} for the sake of comparison. 
For the sake of clarity, we outlay in the main text the derivation of FISTA, and leave the derivation of the alternative approaches to \ref{appendix:ADMM} and \ref{appendix:linearized}.

\xhdr{Algorithm} Broadly speaking, FISTA \cite{beck2009fastgen} is an algorithm for efficiently solving optimization problems of the form: $\min_{x} f(x) +g(x)$,  where $g$ is proper convex (but not necessarily smooth, as typically for $\ell_1$ penalties and indicator set functions) and  the subgradients of $f$ are Lipschitz. 
One of the most appealing characteristics of FISTA lies in (a) the absence of any user-defined parameters---making it a completely parameter-free method---and (b) a $1/k^2$-accelerated convergence rate. In the spirit of the algorithm proposed by Beck and Teboulle \cite{beck2009fast} for image denoising and deblurring under total-variation penalty, we propose to solve  our similarity-based convex problem \ref{eq:sim_CC_fin} using FISTA on the dual. The additional challenges that our approach faces with respect to the original method are two-fold: (a) our set of constraints is given by the graph adjacency matrix $K$ and is thus more general than the regular 2D-grid in \cite{beck2009fast}, and (b) we are optimizing over the set of doubly stochastic matrices, thus requiring an efficient projection algorithm.

\xhdr{From primal to dual} As in the previous section, we begin by supposing that the similarity matrix  $K$ is positive semi definite. $K$ factorizes as: $K=\Phi^T\Phi$, where, by writing $K=U \Lambda U^T$ the spectral decomposition of $K$, we have: $\Phi =\Lambda^{1/2}U^T$.
We emphasize that, while we introduce this (potentially computationally expensive) decomposition to highlight the parallel with image deblurring,  we will never have to explicitly compute it.
Eq. \ref{eq:sim_CC_fin} can thus be equivalently re-written as:
\vspace{-0.3cm}
\begin{multline*}
\text{Minimize}_{\pi \in \R^{N \times N}} \frac{1}{2} ||  \Phi \pi -\Phi ||_F^2+\mathbbm{1}_{\pi\in \Delta_N}\\ +\lambda \Big(\a  || \pi \delta^{(K)}||_{2,1}+(1-\a)||  \pi \delta^{(K)}||_1 \Big) \\
\end{multline*} 
\vspace{-1.2cm}

This is akin to an image deblurring problem, with $\pi$ playing the role of the true image, and $\Phi$ the observed image and blurring process. 
 Similarly to Beck and Teboulle, we thus propose to start with the associated image denoising problem, and will generalize to the original deblurring problem in a subsequent step:
  \vspace{-.1cm} 
 \begin{multline}\label{eq:primal_mercer_denoising}
\hspace{-.3cm}  \underset{\pi\in \Delta_N} {\text{Minimize}}\frac{1}{2}||  \pi -\Phi ||_F^2+ \lambda \Big(\a || \pi \delta^{(K)}||_{2,1}+(1-\a) ||  \pi \delta^{(K)}||_1 \Big)
 \end{multline} 
 \vspace{-0.4cm}  
 
\begin{prop}
The dual of Eq. \ref{eq:primal_mercer_denoising} is given by:
\begin{multline}\label{eq:dual}
	\max_{p \in \mathcal{P},q\in \mathcal{Q}} || \Pi_{(\Delta_N)^C}\Big(\Phi -\lambda \big(  \alpha   p\delta_K^T+  (1-\alpha)  q\delta_K^T \big) \Big)||_F^2\\ -||  \Phi -\lambda \big(  \alpha   p\delta_K^T+  (1-\alpha)  q\delta_K^T  ||_F^2
\end{multline}
where we denote as $\Pi_{\Delta_N}$ the orthogonal projection operator on the set $\Delta_N$ and  $\Pi_{\Delta_N^C}=I-\Pi_{\Delta_N}$ the projection onto its complement, and where the sets $\mathcal{P}$ and  $\mathcal{Q}$ are respectively the $\ell_2$- sphere and the unit cube in $\mathbb{R}^N$:
\begin{equation*}
\begin{split}
\mathcal{P}&=\big\{p \in \R^{N \times N^2} : \pt i, j \in [1,N]^2, \q  ||p_{\cdot,ij}||_2 \leq 1   \big \}\\
\text{ and   } \mathcal{Q}&=\big\{q \in \R^{N \times N^2}  :  \forall i,j \in [1,N]^2, \quad ||q_{\cdot,ij}||_{\oo}\leq 1   \big \}
\end{split}
\end{equation*}
The subgradients associated to this objective are Lipschitz with constant $L=16\lambda^2\max_{i}||K^2_i||_{2}$.
\end{prop}

\begin{proof}We begin by observing that:
 \vspace{-0.5cm} 
 
	\begin{equation*}\label{eq:dual_eq_expression}
\hspace{-0.4cm}	\max_{ p \in  \mathbb{R}^N: ||p||_2 \leq 1} p^Tx =\sqrt{\sum_{i=1}^n x_i^2}
 \text{ and } \max_{ q \in \mathbb{R}^N: ||q||_{\oo } \leq 1} q^Tx =|| x||_1. \end{equation*}
The derivation of these observations is quite simple and given in \ref{appendix:further_proofs} of the extended version of this paper. This allows Eq.\ref{eq:sim_CC_fin} to be re-written as:
 \vspace{-0.6cm} 
 
\begin{multline*}
 \min_{\pi  \in \Delta_N} || \pi-\Phi||_{F}^2+2 \lambda \max_{p  \in \mathcal{P},q \in \mathcal{Q}}\text{Trace}\Big(  \alpha  p^T\pi \delta_K+  (1-\alpha) q^T\pi \delta_K\Big)
\end{multline*}
The corresponding dual problem  $h(p,q)$ is thus given by:
\vspace{-0.5cm}

\hspace{-1cm} \begin{equation*}\label{eq:dual}
\begin{split}
\hspace{-0.cm} &	\max_{p \in \mathcal{P},q\in \mathcal{Q}} \min_{\pi  \in \Delta_N} || \pi-\Phi||_{F}^2+2 \lambda \text{Trace}\Big(  \alpha \delta_K  p^T\pi +  (1-\alpha)\delta_K  q^T\pi \Big)\\
\hspace{-0.cm}  &=	\max_{p \in \mathcal{P},q\in \mathcal{Q}} \min_{ \pi \in \Delta_N} ||\pi - \Big(\Phi-\lambda \big(  \alpha  p\delta_K^T+  (1-\alpha)  q\delta_K^T\big) \Big) ||_{F}^2\\
\hspace{-0.cm} &- ||\Phi -\lambda \Big(  \alpha  p\delta_K^T+  (1-\alpha)  q\delta_K^T\Big)||_F^2\\
	\end{split}
\end{equation*}

The inner expression here is minimized by the projection  of $\Phi -\lambda \Big(  \alpha  p\delta_K^T+  (1-\alpha)  q\delta_K^T \Big)$ onto the set $\Delta_N$, allowing an explicit formulation of the dual as $\max_{p \in \mathcal{P},q\in \mathcal{Q}}  h(p,q)$, with:
\begin{multline*} h(p,q) =|| \Pi_{\Delta_N^C}\Big(\Phi -\lambda \big(  \alpha  p\delta_K^T+  (1-\alpha)  q\delta_K^T  \big) \Big)||_F^2\\-||  \Phi -\lambda \big(  \alpha  p\delta_K^T+  (1-\alpha)  q\delta_K^T \big) ||_F^2
\end{multline*}
which concludes the first part of the proof.

We now have to prove that the subgradients of the dual $h(p,q)$ are Lipschitz. Taking derivatives with respect to $p$ and $q$, we can show that $h$ is Lipschitz with constant:
$$ L(h) = 16 \lambda^2  \max[\a^2,(1-\a)^2] \ \times (\max_{i}||K_i||^2_{2}) $$
We defer the proof to  \ref{appendix:lipschitz} $\qedsymbol$\\
\end{proof}

\noindent This proposition lays the grounds for using accelerated ascent algorithms such as FISTA on the dual: given that we have shown that the subgradients of dual are Lipschitz, FISTA ensures to solve the objective with a convergence rate in $O(\frac{1}{k^2})$, where $k$ denotes the number of iterations \cite{beck2009fastgen}. 
However, as for image deblurring, our setting is further complicated by the presence of the  ``blurring" matrix $\Phi$. While we have assumed here $K$ to be positive definite and could potentially solve exactly the projection update(*), this update would in particular require the inversion of the operator $\Phi^T\Phi=K$ --- a costly operation that does not transfer well in the case where $K$ is nearly singular. Instead, we adopt the approximate strategy of Beck and Teboulle, and view it as a rough equivalent to their deblurring problem. Denoting the solution of the denoising problem by $D(\Phi, \lambda)$, the authors show the optimal solution of the deblurring problem can be obtained by iteratively solving:
$$D( Y -\frac{2}{L}\Phi^T( \Phi \pi - \Phi)), \frac{2\lambda}{L}) =D( Y -\frac{2}{L}( K \pi - K )), \frac{2\lambda}{L}) $$
Empirical results (section \ref{sec:real_exp}) validate our approach.

The FISTA updates of the dual variables are described in Algorithm \ref{alg:FISTA_dual}.

\begin{algorithm}[h]\label{alg:FISTA_dual}
\begin{algorithmic}
	  \STATE {\bfseries Input: (fixed) variables  $\pi^0$, $K$}
	   \STATE{\bfseries Output: Denoising problem output }
	    \STATE{\textit{Initialization: $(p,q) =(s_0,r_0) = (\mathbf{0}_{N \times N^2}, \mathbf{0}_{N \times N^2})$}}
	    	\WHILE{not converged}{
	\STATE{$(p_k,q_k) = \Pi_{\mathcal{P},\mathcal{Q}}\Big[   r_k   +\frac{2\lambda (\alpha,1-\a)}{L(h)} \pi_k \delta_K\Big]  $}
			\STATE{Or equivalently: }
			\STATE{$p_k = \Pi_{\mathcal{P}}\Big[   r_k   +\frac{ \alpha}{8 \lambda  \max[\a^2,(1-\a)^2] \ \times (\max_{i}||K_i||^2_{2})} \pi_k \delta_K\Big]  $}
				\STATE{$q_k = \Pi_{\mathcal{Q}}\Big[   s_k   +\frac{ 1-\alpha}{8 \lambda  \max[\a^2,(1-\a)^2] \ \times (\max_{i}||K_i||^2_{2})} \pi_k\delta_K\Big]  $}
				\STATE{$t_{k+1} = \frac{1 + \sqrt{1+ 4 t_k^2}}{2}$}
				\STATE{$(r_{k+1},s_{k+1}) =(p_k, q_k) + \frac{t_k -1}{t_{k+1}} (p_k -p_{k-1}, q_k - q_{k-1})$}
				\STATE{$\pi_{k+1} =\Pi_{\Delta_N}\big[\pi^k -\frac{2}{L}(K \pi^k -K) - (\a \lambda  r_k \delta_K^T+(1-\a) \lambda  s_k \delta_K^T )  \big] \delta_K\Big] $}
				
	}
	\ENDWHILE
\end{algorithmic}
	\caption{Update for $\pi$}\label{alg:FISTA_dual}
\end{algorithm}

\begin{prop}
	The projection operators onto the sets  $\mathcal{P}, \mathcal{Q}$ are given by:\\
	$\bullet \quad \Pi_{\mathcal{P}}[p] =\frac{ p_{k}}{\max[1,||p||_2]}$\\
	$\bullet \quad \Pi_{\mathcal{Q}}[q] =\frac{ q_{k}}{\max[1,|q_k|]}$\\
\end{prop}
	
In particular, the previous two-step procedure has the advantage of bypassing the need to explicitly compute and invert $\Phi$. In order to efficiently perform the updates on the set of doubly stochastic matrices, we use the scalable iterative scheme proposed in \cite{lu2016fast}, which we detail in Algorithm \ref{alg:FISTA_proj}.

The procedure is summarized in Algorithms \ref{alg:FISTA_dual} and \ref{alg:FISTA_proj}, and a Python implementation is provided on Github\footnote{https://github.com/donnate/HC\_dev}.

\begin{algorithm}
\begin{algorithmic}
	\STATE{$Y $matrix to project onto $\Delta_N$}
	\STATE{$P^*= \text{arg min}_{D \in \Delta_N} || Y-D||_F^2$}
		\STATE{\textit{Initialization: $P\leftarrow Y$;}}
	\WHILE{not converged}{
		\STATE{$P \leftarrow P+ ( \frac{1}{n}I +\frac{\mathbf{1}^TP\mathbf{1}  }{n^2} I -\frac{1}{n} P) \mathbf{1}\mathbf{1}^T -\frac{1}{n}11^TP$}
		\STATE{$P \leftarrow \frac{P+|P|}{2}$}
}
\ENDWHILE
\STATE{Return $\pi^*$ such that $\Phi\pi^*=\Pi_{\Phi \Delta_N}[\Phi-\lambda \mathcal{L}(r_k,s_k)]]$}
\end{algorithmic}
	\caption{Projection onto $\Delta_N$}\label{alg:FISTA_proj}
\end{algorithm}

\vspace{-0.1cm}
\section{Performance Analysis} \label{sec:synthetic_exp}
\xhdr{Computational cost analysis} An inspection of the updates in Algorithm  \ref{alg:FISTA_dual} reveals that only a subset of the coordinates of $p_{ij}$ have to be updated at each iteration--- that is, whenever $K_{ij}=0$, $p_{ij}$ remains identically zero. Denoting $|\mathcal{E}|$ as the number of non-zero entries in $K$, the memory required to store both $p$ and $q$ is thus $O(N|\mathcal{E}|)$.
At each step, the algorithm thus relies on either (a) element-wise operations on matrices of size $O(N|\mathcal{E}|)$  or (b) matrix multiplications with cost at most $O(|\mathcal{E}|N^2)$. As such, the overall complexity and memory storage of the algorithm grows linearly with the number of edges, but quadratically with the number of nodes. We have not as of yet optimized the design of an algorithm capable of efficiently storing a representation of $K$ and leave this to future work (see discussion in the conclusion).\\

\xhdr{Validation of the empirical efficiency}
We begin by assessing the efficiency of our algorithm through a set of synthetic experiments. We generate a synthetic random graph with 3-level fractal structure: the coarsest level corresponds to an Erd{\H o}s-R\'enyi  graph on 4 ``meta nodes". Each of these meta nodes can be further divided in a set of communities on 7 ``super nodes", each corresponding to a dense clique on 7 nodes. This graph generation process yields a graph with 2 levels of clustering (i.e, levels of resolution): a coarse one at the meta level (4 clusters) and a fine-grain one at the super-node level (28 clusters). Figure \ref{fig:hierarchy}(A) illustrates the generation process.  Our convex clustering objective in Eq. \ref{eq:sim_CC} makes it particularly amenable to classification: the columns of the recovered matrix $\pi(\lambda)$ provide a representation of the centroids using the observations as dictionary -- allowing to use any off-the-shelf machine learning algorithm to analyze these representations. We run our hierarchical method on the regularized 2-hop adjacency matrix of the induced graph\footnote{A more detailed explanation is provided in the extended version of this paper: \url{https://arxiv.org/abs/1911.03417}. }  ($K=\text{D}^{-1/2} A^2\text{D}^{-1/2}$ with $D =\text{Diag}( A^2 \mathbf{1})$) and assess the results both visually through the associated PCA plots (Figure \ref{fig:hierarchy}) and quantitatively by running $k$-means on the recovered centroids $\pi(\lambda)$. \\
\textbf{Performance Metrics.} We quantify the amount of structure recovered at each regularization level $\lambda$ through:
\begin{itemize}
\setlength\itemsep{0em}
\item \textbf{the effective rank $er(\lambda)$ \cite{roy2007effective} of the similarity between centroids:} letting $D_{\pi}$ be the distance matrix between observations (i.e., $D_{\pi}[i,j] =  \pi_i^TK\pi_j$) and $\{ \sigma^{(D_{\pi})}\}_j$ its eigenvalues, the effective rank is defined as: $$er(\pi) = \exp\{ -\sum_{k=1}^N \frac{\sigma^{(D_{\pi})}_k }{\sum_{j=1}^N\sigma^{(D_{\pi})}_j}\log\big( \frac{\sigma^{(D_{\pi})}_k }{\sum_{j=1}^N\sigma^{(D_{\pi})}_j}\big)\}.$$  This measures the entropy of the eigenvalue distribution of the similarities between centroids, and should progressively decrease from $N$ to $1$ as $\lambda$ increases. 
\item \textbf{the $k$-means cluster accuracy and silhouette score.} We run $k$-means on the centroids for respectively 4 and 28 clusters, and assess the accuracy of the recovered multilevel clustering: a high accuracy and silhouette score indicate that the convex clustering algorithm has successfully recovered the multi scale structure of the data.
\end{itemize}

\noindent All the results that we show here are averaged over 20 different random graphs. Figure \ref{fig:hierarchy} shows the progressive coalescence of the centroids as $\lambda$ increases. We note the consistency of the recovered cluster path with hierarchical clustering: for very small values of $\lambda$, each cluster contains only one node, and the centroids representations progressively merge as $\lambda$ increases. This can be further quantified by computing the effective rank (Fig. \ref{fig:results}B), which progressively dwindles with the increase of the regularization penalty. Note that, as indicated above, this decline is not strictly monotonic.  This behavior is further quantified in Table \ref{tab:kmeans}, where it becomes apparent that the different coarsened graph representations induced by $\lambda$ recover different levels of resolution: the accuracy both at the meta-community (4 clusters) and super-node (28 clusters) level is extremely high. Yet, as $\lambda$ increases and the centroids progressively fuse, the silhouette score decreases. In particular, Fig.  \ref{fig:results} (C,D)  show that the optimal silhouette score for clustering at the fine and coarse levels shifts from$\lambda^*_{\text{fine}}\approx 1$ to  $\lambda^*_{\text{coarse}}\approx e^{6}$ (peak of the curve).

\begin{figure}[h!]
    \centering
    \includegraphics[width=0.5\textwidth, height=6cm]{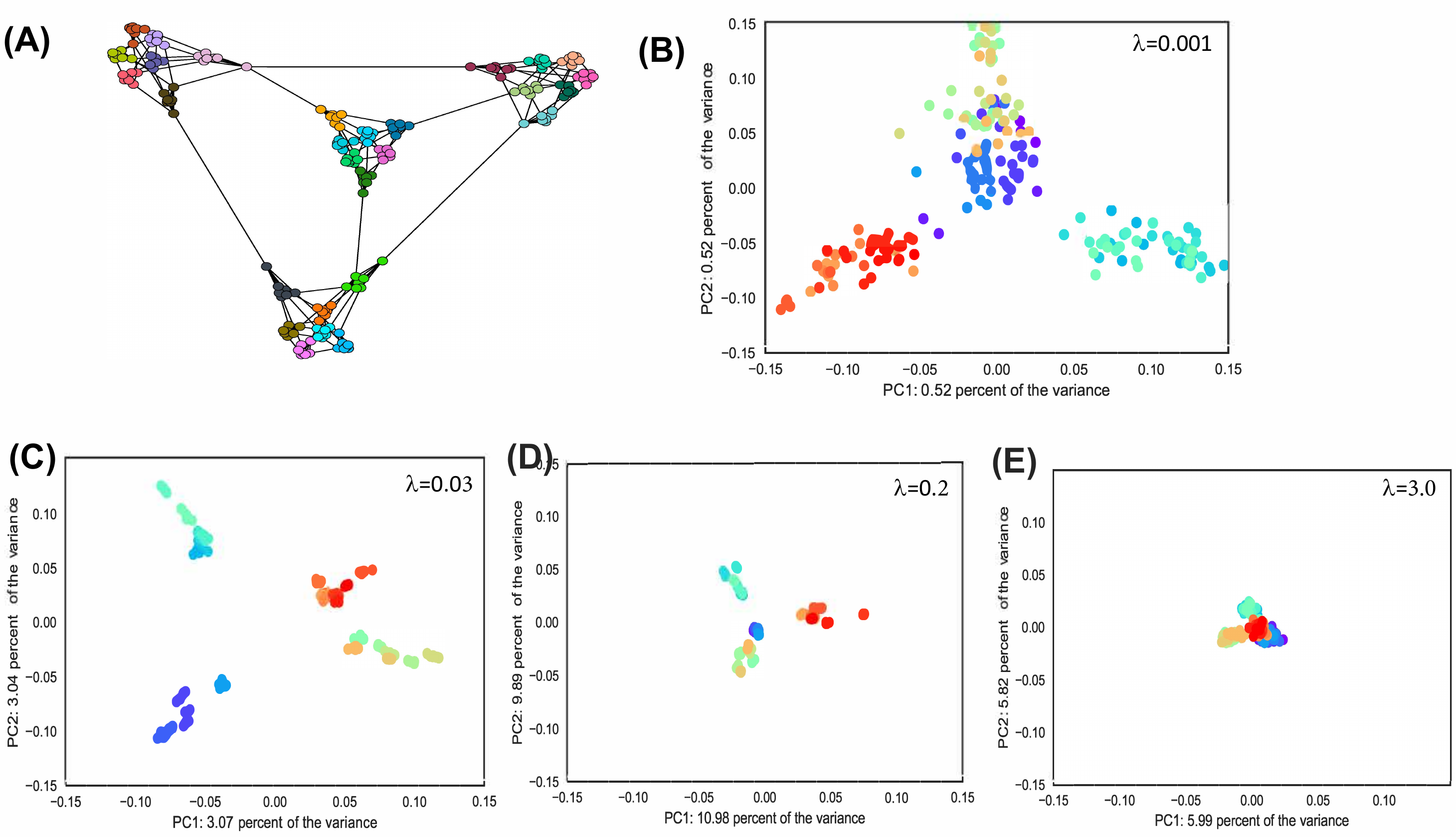}
    \caption{Application of the convex hierarchical clustering algorithm to a synthetic graph on 196 nodes. PCA representation of the nodes for  \textbf{(B)}  $\lambda=0.001$, \textbf{(C)} $\lambda=0.03$, \textbf{(D)} $\lambda=0.2$ and  \textbf{(E) $\lambda=3.0$}.}
    \label{fig:hierarchy}
\end{figure}

\begin{figure}[h!]
    \centering
    \includegraphics[width=0.5\textwidth, height=6cm]{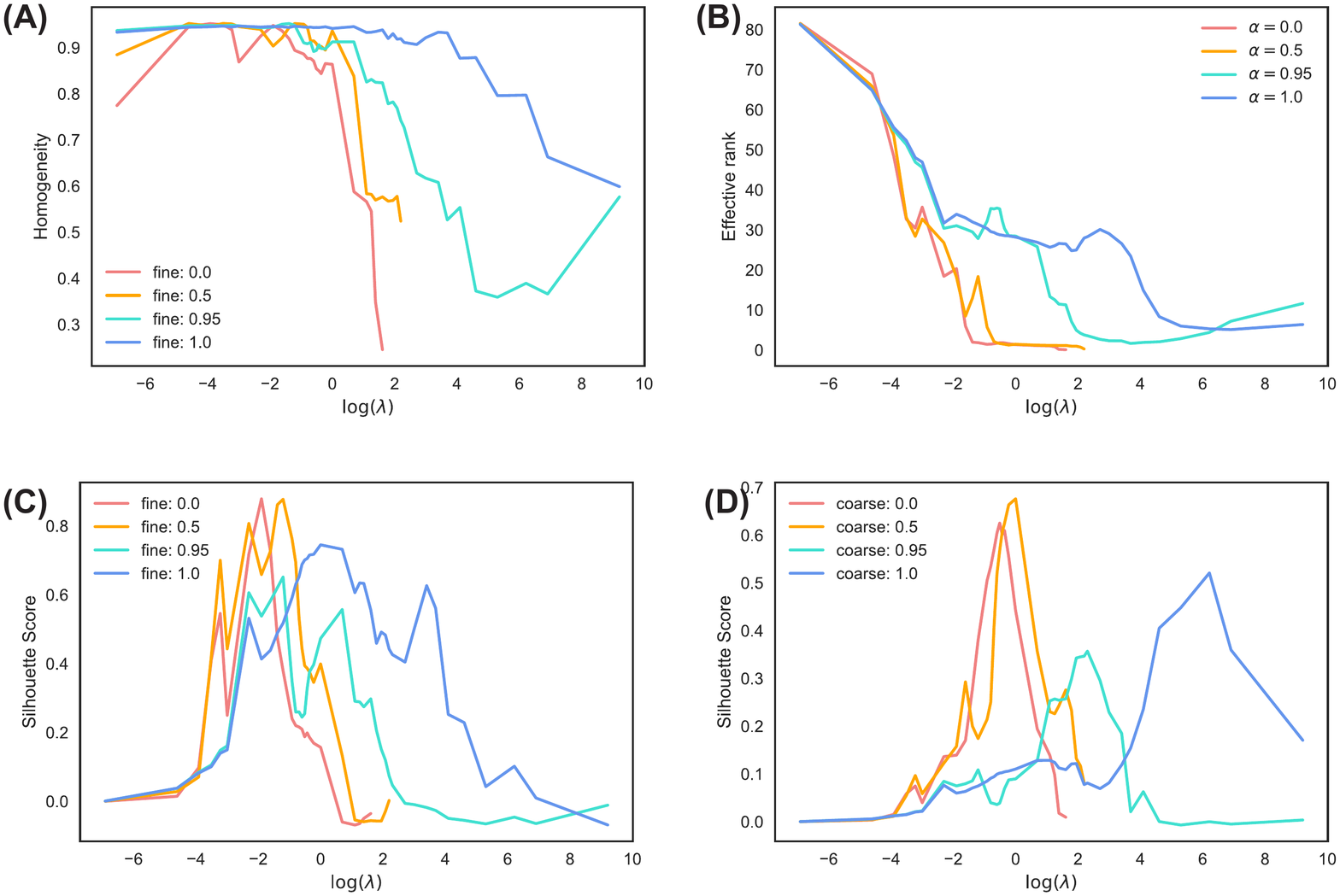}
    \caption{Results of the algorithm averaged over 20 independent trials: (A) efficient rank, (B) homogeneity of the clustering on 28 clusters, (C,D) silhouette scores on respectively 28 and 4 clusters for different values of $\alpha$.} \label{fig:results}
\end{figure}

\begin{table*}[h!]
\centering
  \begin{tabular}{|c|c|c|c|c|c|c|c|}
    \hline
      & &\multicolumn{3}{c|}{ \cellcolor{gray!25} \bf 28 classes} &\multicolumn{3}{c|}{ \cellcolor{gray!25} \bf 4 classes} \\ \hline
   \cellcolor{gray!05}  $\lambda$ &{\small  \cellcolor{gray!05}  $er(\lambda)$}  & {  \cellcolor{gray!05} \small Accuracy}&{\small  \cellcolor{gray!05}  Completeness}& {\small  \cellcolor{gray!05}  Silhouette }  & {  \cellcolor{gray!05} \small Accuracy} &{\small  \cellcolor{gray!05}  Completeness}& {\small  \cellcolor{gray!05}  Silhouette }\\ \hline
    \cellcolor{gray!05}  0.001 & 81.3 & 0.86 & 0.94 & $8.1e^{-4}$ &  0.95 & 0.94 & $1.3e^{-4}$ \\ \hline
  \cellcolor{gray!05}   0.1 & 30.4 &0.89 & 0.94 & $6.1e^{-1}$& 0.95 & 0.95 &$8.5e^{-2}$  \\ \hline
    \cellcolor{gray!05} 0.5 & 35.2 & 0.70 &0.91 &$2.6e^{-1}$ & 0.95 &0.94 &$3.8e^{-2}$ \\ \hline
    \cellcolor{gray!05} 1 &  28.45 & 0.65 &0.91 &$4.7e^{-1}$ & 0.95 &0.95 &$8.9e^{-2}$  \\ \hline
    \cellcolor{gray!05} 40 & 1.64&  0.29 & 0.53  &$-2.6e^{-2}$ & 0.87 & 0.74 &$2.1e^{-2}$ \\ \hline
  \end{tabular}
    \caption{Performance of  kmeans clustering on the raw embeddings, with respectively 4 or 28 classes as ground truth labels, for $\alpha=0.95$. } \label{tab:kmeans}
\end{table*}

\xhdr{Impact of the choice of $\alpha$} A comparison of the results for different values of $\alpha$ is presented in Fig. \ref{fig:results}. We observe that the behavior of the results is roughly similar, although smaller values of $\lambda$ seem to encourage faster clamping of the centroids and a steeper convergence towards the consensus matrix $\pi^{\infty}=\frac{1}{N} \mathbf{1}\mathbf{1}^T$.

\vspace{-0.1cm}

\section{Real-life Experiments}\label{sec:real_exp}
\vspace{-0.1cm}

Our method was driven by its application to multi-resolution graph analysis. In this setting, a typical goal is to obtain a coarser and coarser approximation of the similarity matrix (progressively fusing the clusters together) in order to capture the underlying structure of the graph at multiple scales.
With this objective in mind,  we now provide a few examples of the performance of our method on two real datasets.\\

\hspace{-0.3cm}\textbf{Connectomics.} In this application, we wish to compare the structural connectomes of healthy individuals, undergoing a longitudinal test-retest Reliability and Dynamical Resting-State fMRI study. The data is a subset of the HNU1 cohort \cite{zuo2014open}. In particular, we focus on the structural connectomes of 5 subjects obtained over the course of 10 distinct scan sessions (three days apart from another)\footnote{The preprocessed structural connectomes are readily available at \url{https://neurodata.io/mri-cloud/}.}. For each connectome, we compute its convex clustering representation for various values of the parameter $\lambda$. This yields a multiscale representation from the raw connectomes, which we then compare. The goal is to assess whether the multiscale representations obtained via Convex Clustering are more  consistent and robust across subjects and scans than the ones obtained via traditional single linkage Hierarchical Clustering (HC). To compare the output of our convex clustering procedure (i.e a set of centroids) and the dendrogram obtained via single linkage HC, we compare the distance matrices between centroids that these output induce (in particular, we use the cophenetic distance \cite{sneath1973numerical} to convert the HC dendrogram into a distance matrix). 

Fig. \ref{fig:connectomics} shows, on the left side, the Kendall rank correlation between these similarity matrices for two values of $\lambda$, as well as their correlation with the cophenetic distance induced by single linkage HC.  Interestingly, both HC and Convex clustering recovered multiscale representations with a strong subject effect, as highlighted by the red blocks along the diagonal: representations corresponding to different scans of the same subject are more alike than scans across subjects. This is highlighted by the column on the right side of Fig. \ref{fig:connectomics}, which shows a clear separation in the distances between scans belonging to the same subject (``within distances") and scans across different subjects (``between"). This effect fades away as the regularization increases. This is consistent with our expectation that the overall organization of the brain is globally the same across subjects, while differences between individuals are more salient at the fine-grain scale.

To quantify the relative performance of our algorithm with standard HC, we estimate the variability of its output: for each scan and each value of $\lambda$, we compute the 5 nearest-neighbor graphs that the coarsened similarity matrices induce. We then compute the distances between these 5 nearest-neighbor graphs (using the Hamming distance, that is, the raw $\ell_2$-distance between adjacency matrices). We observe that the variability in these graph is smaller for convex clustering than for graphs obtained using single linkage HC: the distribution for two values of $\lambda$ are plotted on Fig. \ref{fig:connectomics2}, and we observe that these differences are significantly inferior than to the ones of HC. This indicates that the 5-nearest neighbor graphs recovered by our convex clustering procedure induce more robust and consistent multiscale representations of the connectomes across subjects and scans.

\begin{figure}[!h]
\centering
\begin{subfigure}[b]{0.45\textwidth}
\includegraphics[width=\textwidth, height=4.6cm]{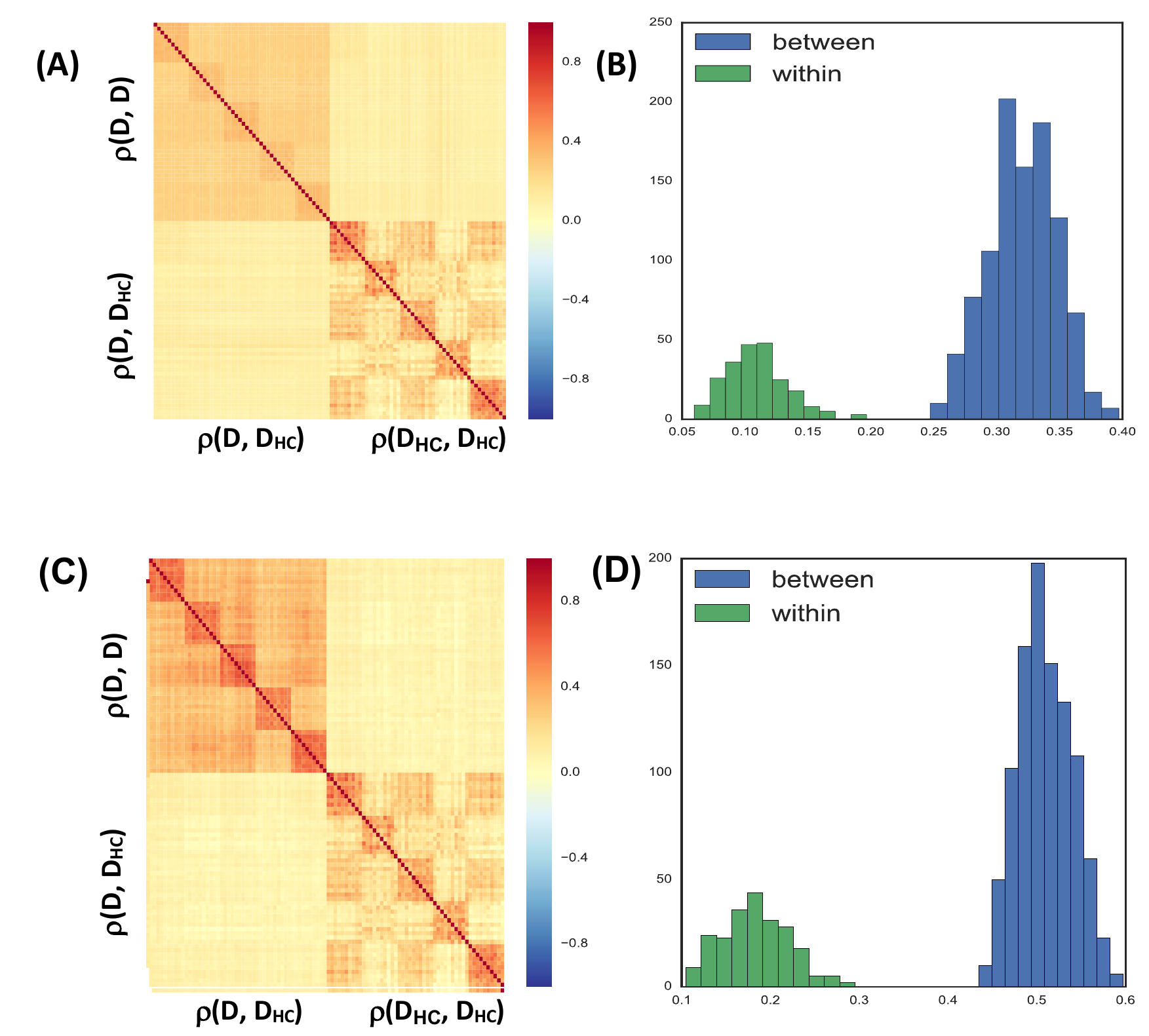}
\caption{Results for the connectomics study using different values of the regularization (A, B: $\lambda=0$ and C,D: $\lambda=0.01$). Left-column: pairwise Kendall rank correlation between coarsened brain networks induced by convex clustering $\rho(D, D)$  and HC's associated cophenetic distance $\rho(D_{HC}, D_{HC})$ as well as Kendall cross-correlation between representations $\rho(D,  D_{HC})$ (matrix sketch on the left of the picture). Right column: distribution of the distances between coarsened DTI scan representations for scans belonging to the same subject ("within") and across different subjects ("between"). }\label{fig:connectomics}
	\end{subfigure}
\begin{subfigure}[b]{0.45\textwidth}
\includegraphics[width=\textwidth, height=5.cm]{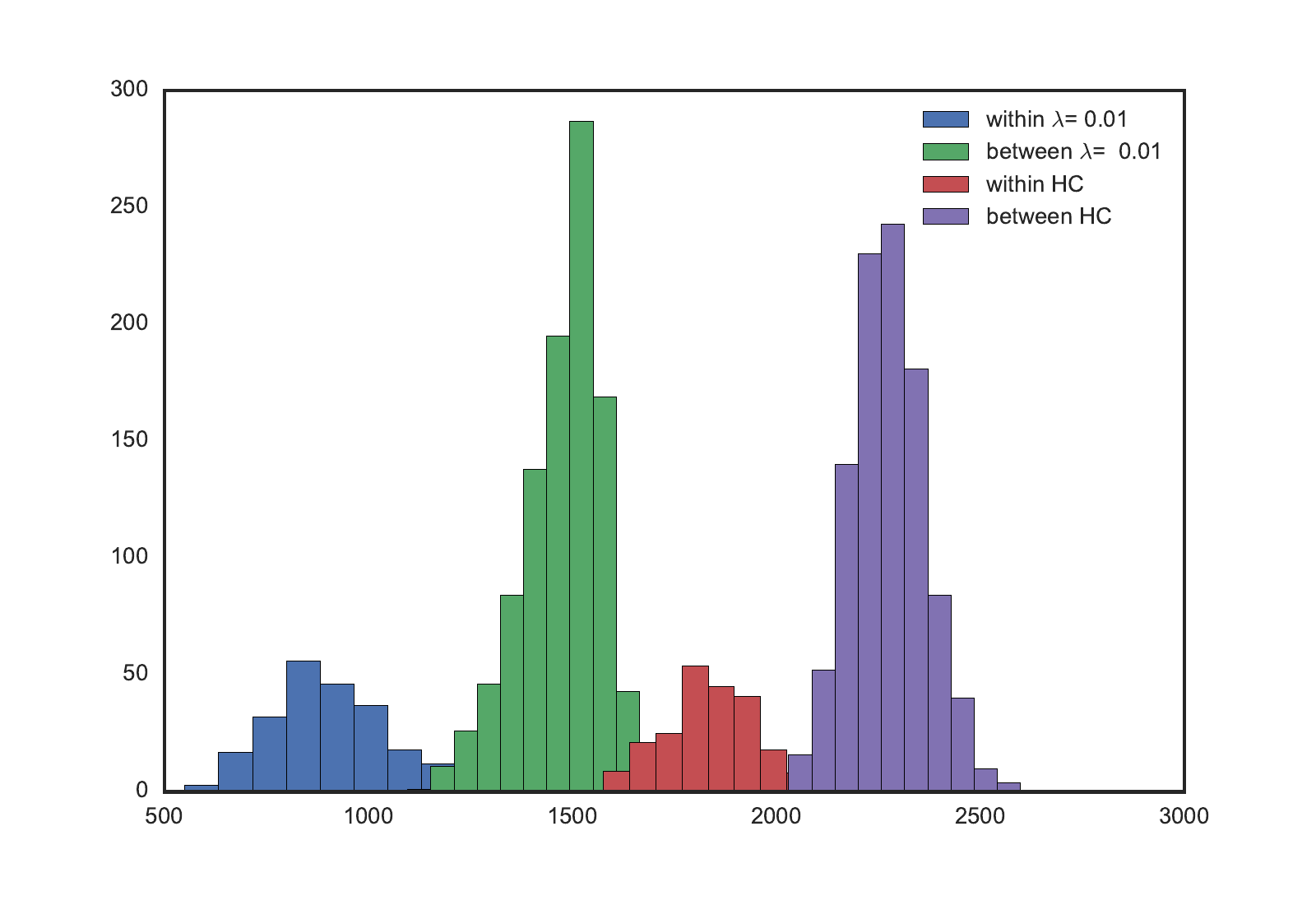}
\caption{Comparison of the "within" and "between" subject distances between coarsened representations of DWI scans for various levels of $\lambda$. }\label{fig:connectomics2}
	\end{subfigure}	
\end{figure}

\textbf{Khan gene expression.} We now demonstrate the robustness of our method by applying it to the Khan dataset \cite{khan2001classification}. This dataset consists of gene expression profiles of four types of cell tumors of childhood. In this case, we want to show that the clusters recovered by our procedure are more robust than those recovered by standard hierarchical clustering, in that the multiscale representation of the similarities between genes that they capture are more reproducible: we split the dataset between training and testing, and assess the similarity between the multi scale representations that we extract out of those.  In this case, a set of 64 arrays and 306 gene expression values are used for training, and 25 arrays for testing. We apply hierarchical clustering on the similarity matrix induced by the data's 10 nearest-neighbor graph and assess the stability of the induced hierarchy: at each level, we aggregate the centroids in both training and testing based on their efficient rank and  compute the clusters' homogeneity score using the training labels as ground truth. We  compare this against standard agglomerative clustering. Interestingly, the results (displayed in the table in Fig. \ref{figurekhan})  indicate a better homogeneity of our method with respect to the greedy one for intermediary values of the regularization, indicating that the more unstable clusters of standard HC's clustering are at the intermediary levels. Fig. \ref{fig:knn_khan} also shows that the distances between train and test k-nn graphs (as in the connectome study) is consistently smaller than for the k-nn graphs induced using HC's cophenetic distance. This indicates greater consistency between test and train results for convex clustering.\\
\vspace{-0.9cm}

	\  \begin{figure}[!ht] \caption{ Results for the Khan Dataset}\label{figurekhan}
		\centering
		\begin{subfigure}[b]{0.5\textwidth}
		\includegraphics[width=7cm, height=4cm]{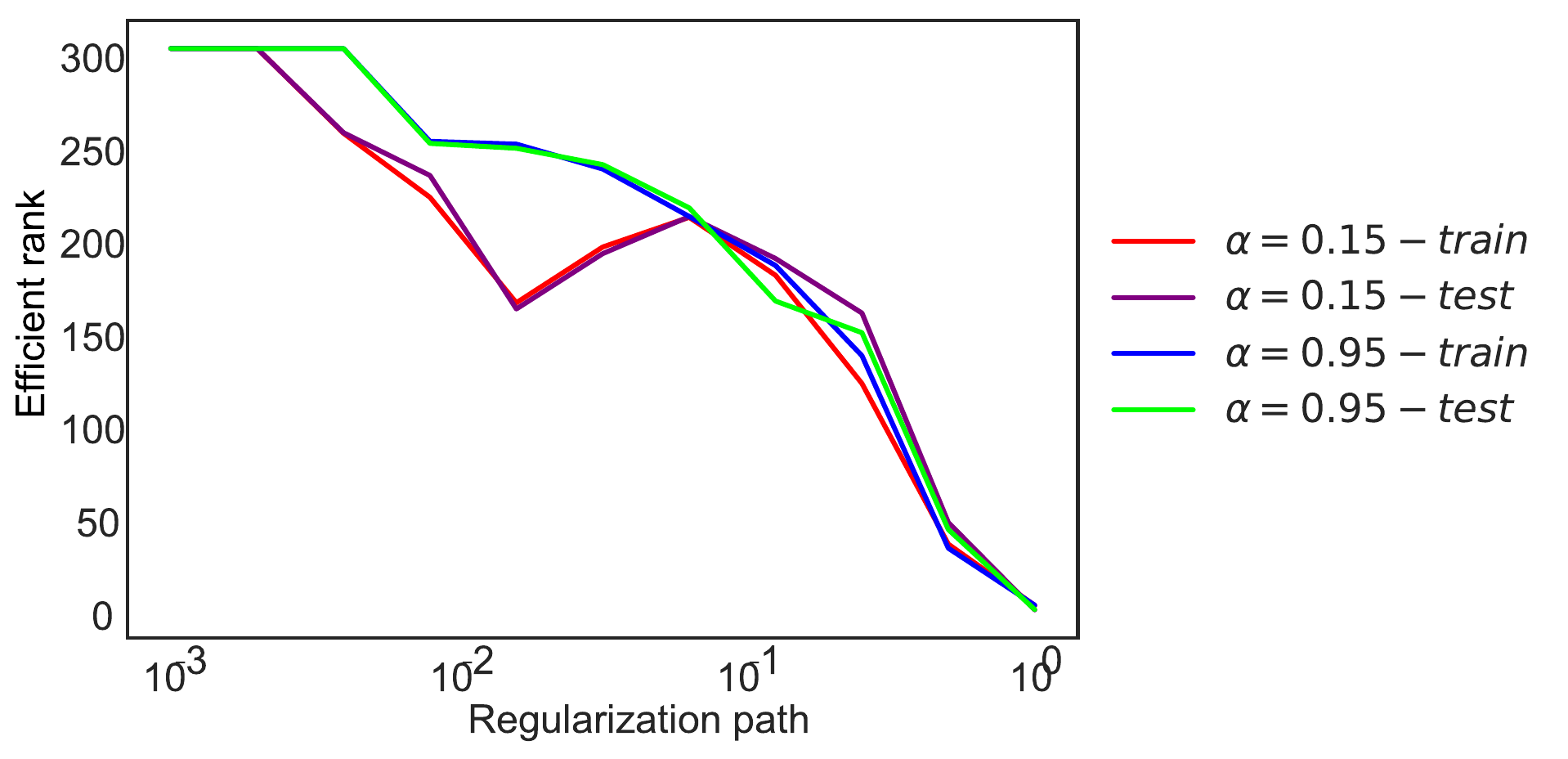}
		\end{subfigure}		\subcaption{Figure: Efficient rank for recovered along the regularization path, on the test and train sets for two values of $\alpha$.}\label{figurekhan_fig}
		\qquad
		\begin{tabular}[b]{|p{0.8cm}|p{1.65cm}|p{2.4cm}|p{2.0cm}|}
			\hline 
			 { ${\lambda}$}	&  { {\bf Effective rk   }  {\scriptsize $er(\pi)$}} &  { \bf  Homogeneity  FISTA} {\scriptsize ($\alpha=0.95$)}&  {  \bf Homogeneity standard HC} \\ 
			\hline 
			0.032& \hspace{0.6cm}  242 & {\hspace{0.6cm}   {\bf0.937}}&  {  \hspace{0.6cm}   0.935}  \\ 
			\hline 
			0.256 &  \hspace{0.6cm} 141 &  \hspace{0.6cm}     \textbf{0.774}  & \hspace{0.6cm}  0.721  \\ 
			\hline 
			0.512 &  \hspace{0.6cm}  37 &   \hspace{0.6cm}  \textbf{0.446 } & \hspace{0.6cm}   0.292  \\ 
			\hline 
			1.024 & \hspace{0.6cm}  7 & \hspace{0.6cm}   \textbf{0.100}  & \hspace{0.6cm}  0.087  \\ 
			\hline 
		\end{tabular}
		\subcaption{ Table: Homogeneity score between test and train predictions at different points of the regularization path.}\label{figurekhan_table}
	\end{figure}
	
		\begin{figure}[h!]
		\centering
		\includegraphics[width=0.45\textwidth, height=3.8cm]{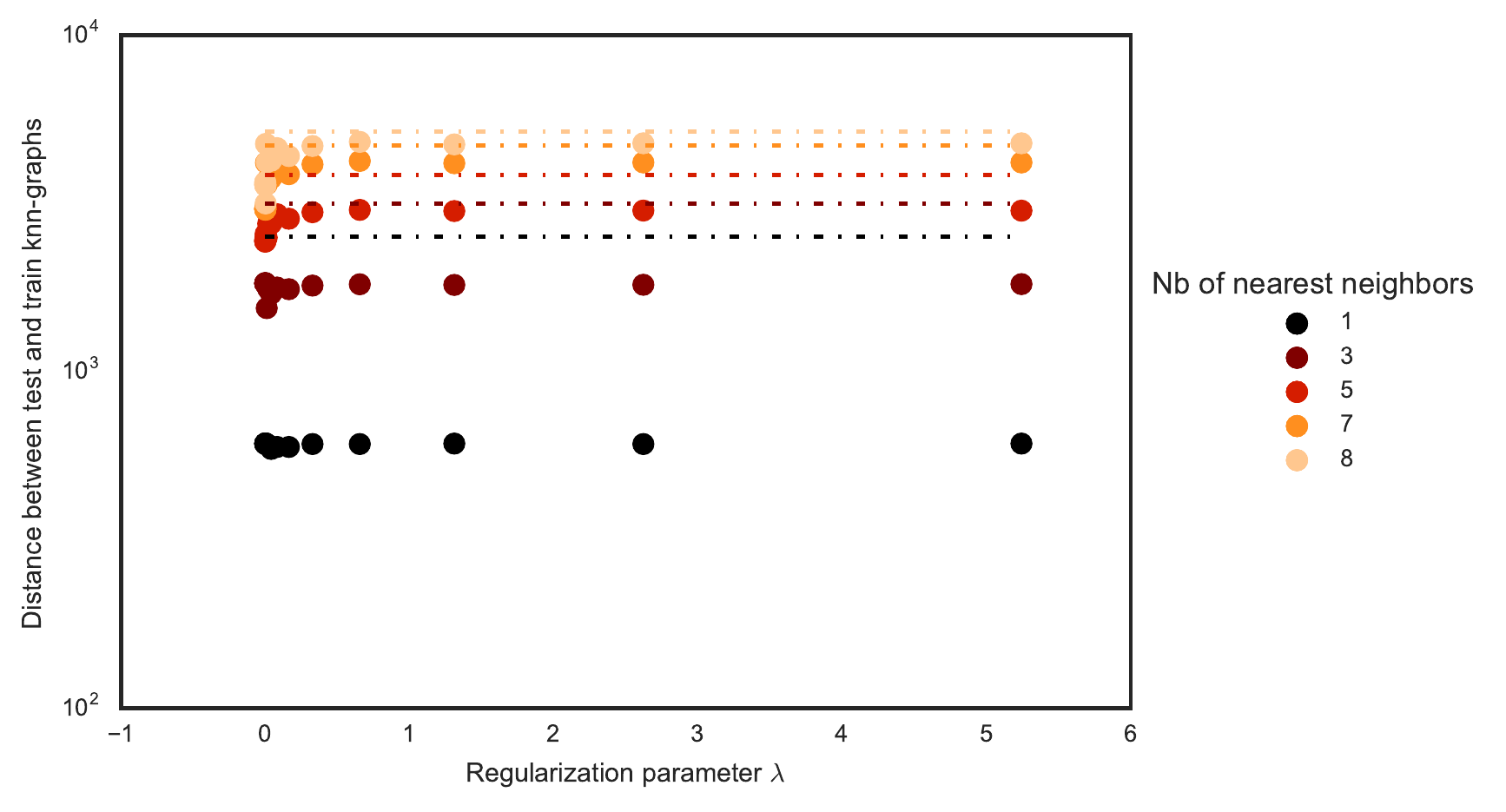}
		\caption{Distances between test and train k-nearest neighbor graphs as the regularization $\lambda$ increases (colored by k). Dashed lines indicate the HC-cophenetic baseline for the number of neighbors considered. }
\label{fig:knn_khan}
\end{figure}

\section{Conclusion}
In conclusion, we have proposed an adaptation of FISTA on the dual for solving convex hierarchical clustering in the case where the data are directly a graph or a similarity matrix. We have shown the performance of our method on both synthetic and real datasets, highlighting its ability to recover different important scales with better consistency and robustness than standard Hierarchical Clustering. To begin working on scaling up this method, we also devised a gradient descent-based implementation, based on a linearization of the objective and more suitable to the analysis of larger graphs, as well as an ADMM version \cite{admm} for the sake of comparison, provided in the extended version of this paper \footnote{An extended version of this paper can be found at: \url{https://arxiv.org/abs/1911.03417}}. One intrinsic limit to the scalability of our method lies in its requirement to store a matrix of size $N^2$ -- an aspect that we leave for future work.  

\xhdr{References}
\vspace{-0.1cm}
{ \tiny
	\bibliographystyle{abbrv}

}
\setcounter{section}{0}
\renewcommand{\thesection}{Appendix \Alph{section}}
\renewcommand{\thesubsection}{\Alph{section}}
\section{Lipschitz Constant}\label{appendix:lipschitz}
\setlength{\abovedisplayskip}{0pt}
\setlength{\belowdisplayskip}{0pt}
\vspace{-0.2cm}

 We here show that the dual problem is Lipschitz with respect to each of the variables $p$ and $q$.
 We have:
 \begin{multline*} h(p,q) =|| \Pi_{\Delta_N^C}\Big(\Phi -\lambda \big(  \alpha  p\delta_K^T+  (1-\alpha)  q\delta_K^T  \big) \Big)||_F^2\\-||  \Phi -\lambda \big(  \alpha  p\delta_K^T+  (1-\alpha)  q\delta_K^T \big) ||_F^2
\end{multline*}
Thus:
 \begin{multline*}
 \nabla_p h(p,q)= -\lambda \a \Big(2\Pi_{\Delta_N^C}[\Phi-\lambda(\a p \delta_K^T +(1-\a) q \delta_K^T)] \\-2 \big( \Phi-\lambda(\a p \delta_K^T +(1-\a) q \delta_K^T)  \big)  \Big)\delta_K \q (*)\\= 2\lambda \a \Pi_{\Delta_N}[\Phi-\lambda(\a p \delta_K^T +(1-\a) q \delta_K^T)] \delta_K
 \end{multline*}
where $(*)$ follows from the fact that $\frac{\partial || \Pi_{\Delta_N^C} [x]||_F^2}{\partial x } =\frac{\partial || \Pi_{\Delta_N^C}[ x]||_F^2}{\partial \Pi_{\Delta_N^C} [x] }\frac{\partial \Pi_{\Delta_N^C} [x]}{\partial x} = 2 \Pi_{\Delta_N^C} [x] $.\\
Similarly:
$$  \nabla_q h(p,q) = 2\lambda (1-\a) \Pi_{\Delta_N}[\Phi-\lambda(\a p \delta_K^T +(1-\a) q \delta_K^T)] \delta_K$$
We note that, by definition of $\delta_K$:
$$ \forall M \in \R^{N \times N^2}, \q ||M\delta_K||_{ij}^2 = K_{ij}^2 || M_i -M_j||_2^2$$
Hence:
\begin{equation}
\begin{split}
\hspace{-0.1cm} & \forall M \in \R^{N \times N^2}, ||M\delta_K||_F^2 \\
&= \sum_{ij} K^2_{ij} || M_i -M_j||^2 \leq \sum_{ij} 2K_{ij}^2 (|| M_i||^2+||M_j||^2)\\
\hspace{-0.1cm} & \leq  \sum_{ij} (2K_{ij}^2 || M_i||^2+2 K_{ji}^2||M_j||^2)  \hspace{0.1cm} \text{ \footnotesize by symmetry of $K$}\\
\hspace{-0.1cm} & \leq 4 \sum_{i } || K_{i, \cdot}||_2^2 || M_i||^2 \leq 4  \max_{i}\{ || K_{i \cdot }||^2_2 \} \times ||M||_F^2\\
\end{split}
\end{equation}
Hence, using the non-expensiveness property of the orthogonal projection operator, we can show that the subgradients of $h$ are Lipschitz, since:
\begin{equation}
\begin{split}
\hspace{-.4cm}&||\nabla h(p_1,q_1)-\nabla h(p_2,q_2)||_F^2\\
&=||\nabla_p h(p_1,q_1)-\nabla_p h(p_2,q_2)||_F^2 \\
&+||\nabla_q h(p_1,q_1)-\nabla_q h(p_2,q_2)||_F^2\\
\leq  &8  \lambda^2  \max[\a^2,(1-\a)^2] \times 4 \times \max_{i}||K_{i \cdot}||^2_{2}\\
& \times||  \Pi_{\Delta_N}\Big(\Phi -\lambda \big(  \alpha p_1\delta_K^T+  (1-\alpha)  q_1\delta_K^T\big) \Big)\\
&-\Pi_{\Delta_N}\Big(\Phi -\lambda \big(  \alpha  p_2\delta_K^T+  (1-\alpha)  q_2\delta_K^T\big) \Big)||_F^2\\
\leq& 32 \lambda^4  \max[\a^2,(1-\a)^2] \ \times (\max_{i}||K_{i \cdot}||^2_{2}) \\
& \times ||  \Big( \alpha  (p_1-p_2)+  (1-\alpha)  (q_1-q_2) \Big) \delta_K^T||_F^2\\
\end{split}
\end{equation}
\begin{equation}
\begin{split}
&||\nabla h(p_1,q_1)-\nabla h(p_2,q_2)||_F^2\\
 \leq & 128\lambda^4   \max[\a^2,(1-\a)^2]  \ \times (\max_{i}||K_{i \cdot}||^2_{2}) \\
 & \times||  \big(\alpha  (p_1-p_2)+  (1-\alpha)  (q_1-q_2)\big)||^2_F\\
\leq& 128\lambda^4   \max[\a^4,(1-\a)^4]  \times (\max_{i}||K_{i \cdot}||^2_{2})\\
&  \times 2(||  p_1-p_2||^2_F+  ||q_1-q_2)\big)||^2_F)\\
\leq& 256 \lambda^4  \max[\a^4,(1-\a)^4] (\max_{i}||K_{i \cdot}||^2_{2})\\
&\times  ||(p_1,q_1)-(p_2,q_2)||^2_F
\end{split}
\end{equation}
Thus:
\begin{multline*}
||\nabla h(p_1,q_1)-\nabla h(p_2,q_2)||_F \leq 16 \lambda^2 \times \\ \max[\a^2,(1-\a)^2] \ \times (\max_{i}||K_{i \cdot}||_{2})   ||(p_1,q_1)-(p_2,q_2)||_F
\end{multline*}

\section{More on the derivation of the FISTA updates}\label{appendix:further_proofs}
\setlength{\abovedisplayskip}{0pt}
\setlength{\belowdisplayskip}{0pt}

In this appendix, we provide more in-depth descriptions of the propositions and proofs derived in this manuscript.\\

\xhdr{Derivation of the adapted convex objective problem} We begin by showing how, in the case where the kernel matrix $K$ is assumed to be positive definite, the adaptation of convex clustering proposed in Eq. \ref{eq:sim_CC} naturally follows. To begin with, we remind the reader that, by Mercer's theorem, we can simply write a high-dimensional equivalent formulation $\hat{\mathcal{P}}$ of convex clustering as:
\begin{multline}
\hat{\mathcal{P}}= \min_{\pi _\in \Delta_N} \sum_{i=1}^N || \Phi(X_i) -\sum_{j}\Phi(X_j) \pi_{ji}||^2 \\+\lambda \sum_{i,j}W_{ij}||\sum_{k}\pi_{ki}\Phi(X_k)-\sum_{k}\pi_{k'j}\Phi(X_{k'})||
\end{multline}
This induces the following set of equivalents:
\begin{equation}  
\begin{split}
\hat{\mathcal{P}} \iff &\text{arg min}_{\pi _\in \mathcal{S}} \sum_{i=1}^n \Big( || \Phi(X_i)||^2 +||\Phi(X) \pi_{\cdot,i}||^2 \\
&-2  \Phi(X_i)^T (\Phi(X) \pi_{\cdot,i})\Big) +\lambda \sum_{i,j}W_{ij}\Big(||U_i -U_j|| \Big)\\
\iff &\text{arg min }_{\pi _\in \mathcal{S}} Tr[ \pi^T \Phi(X)^T \Phi(X)\pi ] \\
&-2Tr \Big(\Phi(X)^T  (\Phi(X) \pi )\Big) +\lambda \sum_{i,j}W_{ij}\Big(||U_i -U_j|| \Big) \\
\iff &\text{arg min }_{\pi _\in \mathcal{S}} Tr [ \pi^T \Phi(X)^T \Phi(X)\pi ] \\
&-2Tr \Big(\Phi(X)^T  (\Phi(X) \pi )\Big)  \\
&+\lambda \sum_{i,j}W_{ij}\underbrace{\Big(||\Phi(X) [\pi_{\cdot i} -\pi_{\cdot j}]||}_{\leq L || \pi_{\cdot i} -\pi_{\cdot j} ||} \Big)\\
\implies &\text{arg min }_{\pi _\in \mathcal{S}} Tr [ \pi^T K \pi ]\\
& -2Tr\Big[K  \pi\Big] +\lambda \sum_{i,j}W_{ij}\Big(||\pi_{\cdot i} -\pi_{\cdot j}|| \Big)\\
\end{split}	
\end{equation}
where, in the last line, we have used the fact that the columns of $U$ are in fact the coordinated of the centroids in the dictionary of the original observations -- hence, penalizing the pairwise differences between euclidean representation is equivalent to penalizing the dictionary coordinates. \\

\xhdr{Proof of statement \ref{eq:dual_eq_expression}} 
We here provide a brief proof of the statement in Eq. \ref{eq:dual_eq_expression}:
	\begin{equation*}
	\max_{ p \in  \mathbb{R}^N: ||p||_2 \leq 1} p^Tx =\sqrt{\sum_{i=1}^n x_i^2}
 \text{ and } \max_{ q \in \mathbb{R}^N: ||q||_{\oo } \leq 1} q^Tx =|| x||_1 \end{equation*}
To see this, let us first consider the equality on $p$, and introduce the Lagrangian corresponding to the constraint:
$$ \mathcal{L}(p, \lambda )= -p^T x +\lambda (p^T p -1), \q \lambda \geq 0$$
where the primal is
$ \min_{p} \max_{\lambda \in \R^+}  \mathcal{L}(p, \lambda ),$ and the dual can be written as:
$ \max_{\lambda \in \R^+} \min_{p}  \mathcal{L}(p, \lambda ).$
The latter inner minimization with respect to $p$ is achieved for:
 \begin{multline*}\nabla_p \mathcal{L}(p, \lambda )= -x +2 \lambda p=0 \iff  p=\frac{1}{2\lambda}x, \end{multline*}and the dual problem reduces to:
$$ \max_{\lambda}  -\frac{||x||^2}{2\lambda} -\lambda(  \frac{1}{4\lambda^2} ||x||^2 -1) =\max_{\lambda} - \frac{||x||^2}{4\lambda}+\lambda$$
The latter is achieved for $\lambda= \frac{||x||}{2}$, and thus: $ p=\frac{1}{||x||}x$.
Hence, $\max_{ p \in \R^n, p^Tp \leq 1} [ p^Tx]=||x||_2$, which concludes the proof.\\
Similarly for q, it is easy to check that: $$ ||x||_1 =\max_{s: \q s_i \in \{-1, 1\}} s^Tx.$$
By relaxing the constraint on $s$, we have:
$||x||_1 =\max_{s: \q s_i \in [-1, 1]} s^Tx, $
which concludes the proof.

\section{Derivation of the ADMM updates}\label{appendix:ADMM}

In this appendix, we provide the derivations of the ADMM algorithm used to benchmark  our FISTA-based approach in section \ref{sec:synthetic_exp}.

\subsubsection{Description of the algorithm}

The Alternating Direction Method of Multipliers \cite{admm} is a popular algorithm for solving convex optimization problems with a large number of constraints. Indeed, with a guaranteed speed of convergence in $O(\frac{1}{k})$ iterations, this algorithm has become the work-horse of convex problems with coupling constraints. However, contrary to the parameter-free implementation of convex clustering with FISTA, ADMM requires the selection of the parameter $\rho$, whose choice has been shown to considerably affect the speed of convergence \cite{admm,wahlberg2012admm}.
In what follows, in order to simplify the notations, denoting as $e$ the vectors of the Cartesian basis, we introduce the pairwise-difference matrix $\delta \in \mathbb{R}^{N \times N^2}: \delta_{k,ij}=\mathbf{e}_{k,i}-\mathbf{e}_{k,j}$
Introducing the variables $Z_{ij}=\pi_i -\pi_j$ and dual variables $u_{ij}$, the ADMM-augmented Lagrangian can be written as:
\begin{equation} \label{eq:ADMM}
\begin{split}
\min_{\pi \in \Delta_N}& \frac{1}{2} \text{Tr}(\pi^T K \pi -2K^T \pi) +\frac{\rho}{2} \sum_{ij}||\pi \delta +u_{ij} -Z_{ij}   ||^2\\
&+\lambda \sum_{ij} K_{ij}(\alpha || Z_{ij} ||_1 +1-\alpha)|| Z_{ij}||_2 )\\
& \text{s. t. }  \pi \in \Delta_N, \quad \forall i,j, \quad \pi_i -\pi_j = \pi \delta_{ij}= Z_{ij}  \\
\end{split}
\end{equation} 
The full algorithm and derivation of the updates are provided in the following subsection and the whole procedure is summarized in Alg. \ref{ADMM}, and the corresponding updates are derived in the following paragraphs.
\begin{algorithm}
\begin{algorithmic}
	\STATE{\textbf{Input:}  Similarity matrix $K$, regularization parameter $\lambda$}
	\STATE{\textbf{Output:} Optimal solution $\pi^{(\lambda)}$ }
	\STATE{\textit{Initialization;}   $Z, U=\mathbf{0} \in \mathbb{R}^{N \times N^2}, t=0$}
	\WHILE{not converged}{
		\STATE{$\pi^{t+1}=\text{Update}_{\pi}(Z^t, U^t)$   \COMMENT{explicited in Algorithm \ref{fista4ADMM}}}
		\STATE{$Z^{t+1}\leftarrow  \text{SoftThreshold}_{\frac{\alpha \lambda||Z^t||}{\rho||Z^t||+ (1-\alpha)\lambda}} \Big[ \frac{\pi^{t+1}\delta +U^t}{ (1+\frac{(1-\alpha)\lambda}{\rho||Z^t||})}  \Big] $}
		\STATE{$U^{t+1}\leftarrow U^t +(\pi^{t+1}\delta -Z^{t+1})$}
		\STATE{$t \leftarrow t+1$}
	}
	\ENDWHILE
	\STATE{Return $\pi^*=\Pi_{\Delta_N}(\pi)$}
	\caption{ADMM}
	\label{ADMM}
\end{algorithmic}
\end{algorithm}

\subsubsection{Updates}

\textbf{Updating $\pi$}. The objective in Eq. \ref{eq:ADMM} reads as quadratic linear optimization problem in $\pi$. The updates in $X$ are unfortunately not computable  in closed form. However, provided that we have access to an efficient projection on the set of doubly stochastic matrices $\Delta_N$, we can solve the corresponding update using an accelerated Proximal Descent algorithm. In particular, the gradients  with respect to $\pi$ are given by:
$$ \nabla_{\pi} F(\pi, Z,u) = K \pi -K +\rho (\pi \delta +U-Z)\delta^T $$
We note that these gradients are in particular Lipschitz (with respect to $\pi$, all other variables being fixed):
\begin{equation*}
\begin{split}
\nabla_{\pi} F(\pi_1, Z,u) -\nabla_{\pi} F(\pi_2, Z,u)=K(\pi_1-\pi_2) +\rho(\pi_1 -\pi_2)\delta \delta^T
\end{split}
\end{equation*}
We also have:
 \begin{multline*}
\delta \delta^T = \Big(\sum_{ij} (e_{ki}-e_{kj})(e_{li}-e_{lj}) \Big)_{kl} = 2\Big(\sum_{j}(e_{lk}- e_{lj})\Big)_{kl}\\= 2\Big(n e_{lk}-  e_{ll}\Big)  =2n I -2\mathbf{1}\mathbf{1}^T  
\end{multline*}

\begin{equation}
\begin{split}
\implies & || \delta \delta^T ||^2_F = \text{Trace}[4n^2 I -8n\mathbf{1}\mathbf{1}^T + 4n \mathbf{1}\mathbf{1}^T] \\
\implies & || \delta \delta^T ||^2_F \leq  4n^3 \\
\implies &  || \nabla_{\pi} F(\pi_1, Z,u) -\nabla_{\pi} F(\pi_2, Z,u)||_F \\
& \leq \sqrt{||K||^2_F  +\rho^2 || \delta \delta^T||^2 } || \pi_1-\pi_2||_F  \\
&\leq \sqrt{||K||_F^2  +4\rho^2 n^3 } || \pi_1-\pi_2||_F  \\
\end{split}
\end{equation}
Hence, to solve for $\pi$, we can use an accelerated proximal method (such as FISTA), with constant step-size $L =\sqrt{||K||_F^2  +4\rho^2 n^3 }$. Since, the projection onto $\Delta_N$ does not have a closed form solution either, the literature typically resorts to fixed-point algorithms such as the one proposed in \cite{lu2016fast}, yielding the procedure described in Algorithm \ref{fista4ADMM}.\\

	\begin{algorithm}
	\begin{algorithmic} 
			\STATE{\textbf{Algorithm:} Updates for $\pi^{t}{(\lambda)}$ }
		\STATE{\textbf{Input:}(fixed) variables  $K$, $Z$ and $U$}
		initialization: $t_k=1$\;
		\WHILE{not converged}{
			\STATE{$\pi^k= \Pi^{\Delta_N}_{\text{1-round}} \big(Y_{k}-\frac{1}{\sqrt{||K||_F^2  +4\rho^2 n^3 }}  \nabla_{\pi}F(Y^k, Z^t, U^t) \big)$
			 \text{(Projection $\Pi^{\Delta_N}_{\text{1-round}}$ described in Alg. \ref{proj_deltaN})}}
			\STATE{$t_{k+1} =\frac{1+\sqrt{1+4t_k^2}}{2}$}
			\STATE{$Y_{k+1}\leftarrow \pi_k +\frac{t_k-1}{t_{k+1}}(\pi_k -\pi_{k-1})$}
		}
	        \ENDWHILE
		\caption{Updates for $\pi$.}
		\label{fista4ADMM}
		\end{algorithmic}
	\end{algorithm}
	
		\begin{algorithm}
	\begin{algorithmic} 
		\STATE{\textbf{Objective:} Projection on $\Delta_N$: $\Pi^{\Delta_N}_{\text{1-round}}$} 
		\STATE{\textbf{Input:}square matrix $Y$}
		\STATE{Initialization: $P=Y$}
		\WHILE{not converged}{
		\STATE{$P\leftarrow P+ ( \frac{1}{n}I +\frac{\mathbf{1}^TP\mathbf{1}  }{n^2} I -\frac{1}{n} P) \mathbf{1}\mathbf{1}^T -\frac{1}{n}11^TP$}
		\STATE{$P \leftarrow \frac{P+|P|}{2}$}
		}
		\ENDWHILE
		\STATE{Return $P$}
		\end{algorithmic}
		\caption{One-round of the fixed point iterative algorithm  ($\Pi^{\Delta_N}_{\text{1-round}}$) for projecting unto $\Delta_N$ as proposed in \cite{lu2016fast}.}\label{proj_deltaN}
	\end{algorithm}
	

\textbf{Updating $Z$}. The updates in terms of $Z$ are more explicit, since $Z$ is simply the solution to a denoising problem with an elastic net penalty. 

Taking the gradient with respect to $Z$ yields:
\begin{equation*}
\begin{split}
 &\rho(Z -X\delta -U) +\lambda \alpha \sign{Z} +\lambda (1-\alpha) \frac{Z}{||Z||}=0\\
  \end{split}
\end{equation*}
We solve the later through a set of sequential updates:
\begin{equation*}
\begin{split}
  & (1+\frac{(1-\alpha)\lambda}{\rho || Z^{t-1}||})Z^t = X\delta +U+ \frac{\lambda \alpha}{\rho} \sign{Z^t} \\
 \implies& Z^t = \text{SoftThreshold}_{\frac{\alpha \lambda  || Z^{t-1}||} { \rho || Z^{t-1}||  + (1-\alpha)\lambda}} \Big[ \frac{1}{ (1+\frac{(1-\alpha)\lambda}{\rho  || Z^{t-1}||})}(X\delta +U)  \Big].
 \end{split}
\end{equation*}

\section{Large Scale Computations: derivation of the linearization algorithm}\label{appendix:linearized}
\setlength{\abovedisplayskip}{0pt}
\setlength{\belowdisplayskip}{0pt}
\vspace{-0.2cm}

 We provide here an alternative way of solving the problem using gradient descent, which might be better suited to larger scale problems.
 We begin by reminding that the problem that we are solving has the following form:
 \begin{equation}
 \begin{split}
& \text{argmin}_{\pi} \text{Tr}[\pi^T K \pi - 2 K \pi] + \lambda \sum_{i,j} K_{ij} || \Phi \pi_{\cdot i} - \Phi \pi_{\cdot j}||^2\\
\iff & \text{argmin}_{\pi}   \text{Tr}[\pi^T K \pi - 2 K \pi]  \\
&+ \lambda \sum_{i,j} K_{ij} (|| \Phi \pi_{\cdot i}||^2 + ||\Phi \pi_{\cdot j}||^2 -2  \pi_{\cdot i}^T \phi^T \Phi \pi_{\cdot j})\\
\iff & \text{argmin}_{\pi}   \text{Tr}[\pi^T K \pi - 2 K \pi]  \\
& +2 \lambda\text{Tr} \big[\pi^T K \pi \text{Diag}( \tilde{K} \mathbf{1})\big] - 2   \mathbf{1}^T(\pi^T K \pi \odot \tilde{K})\mathbf{1} \\
 \end{split}
 \end{equation}
 where $\mathbf{1}^T \pi =\mathbf{1} $, and $\tilde{K} = K -\text{diag}(K)$. Let us write $\Delta = \{ \pi \in \mathbb{R}^{n\times n} : \mathbf{1}^T \pi = \mathbf{1} \}$ the space of row-wise stochastic matrices.
To compute a solution, we propose using an iterative algorithm based on a linearization of the previous objective function. Introducing $x$ such that $||x|| \leq \delta$, at each iteration $t$, we update $\pi^t$ as $\pi^{t} = x + \pi^{t-1}$. All we need is thus to solve for the updates $x$ at each iteration. Linearizing the previous equation with respect to $x$ yields:
 \begin{equation} \label{eq:lin}
 \begin{split}
& \text{argmin}_{||x||\leq \delta }   \text{Tr}[2\pi_{t-1}^T K x - 2 K x]  \\
& +2 \lambda \text{Tr} \big[  \text{Diag}( \tilde{K} \mathbf{1}) \pi_{t-1}^T Kx \big] - 2 \lambda   \mathbf{1}^T((\pi_{t-1}^T K x  + x^T K \pi_{t-1} )\odot \tilde{K})\mathbf{1} \\
 \end{split}
 \end{equation}
 such that $||x|| \leq \delta$ and $\pi^t = \pi^{t-1} + x \in \Delta$.
 The gradient with respect to $x$ of the previous objective function is:
 $$ \nabla_x \ell(x, \pi^{t-1}) =  2 K \pi^{t-1} -2K + 2 \lambda    K \pi^{t-1} \text{Diag}( \tilde{K} \mathbf{1}) + 4  \lambda K \pi^{t-1} \tilde{K}.$$
 Denoting as $\Pi_{\mathcal{B}_{\delta}} $ and $\Pi_{\Delta} $ respectively the projections on the ball of radius $\delta$ and on $\Delta$, we can thus use projected gradient descent to solve the previous problem, as described in Alg. \ref{alg:lin}.

 	\begin{algorithm}
	\begin{algorithmic} 
		\STATE{\textbf{Objective:} Solve Eq. \ref{eq:lin}} 
		\STATE{\textbf{Input:} $K, \tilde{K}$ and initialized $\pi^0$}
		\STATE{Initialization: $x=0$}
		\WHILE{not converged}{
		\STATE{$x \leftarrow \Pi_{\mathcal{B}_{\delta}} \big(x - \eta \nabla_x \ell(x, \pi^{t-1})\big)$}
		\STATE{$\pi^{t} \leftarrow  \Pi_{\Delta} (\pi^{t-1} + x)$}
		}
		\ENDWHILE
		\STATE{Return $\pi^t$}
		\end{algorithmic}
		\caption{Linearization algorithm}\label{alg:lin}
	\end{algorithm}
	


\end{document}